\documentclass[11pt]{article}

\usepackage{palatino}
\usepackage{mathpazo}
\usepackage{stmaryrd}

\usepackage{hyperref}
\hypersetup{pdfpagemode=UseNone}

\usepackage{amsfonts}
\usepackage{amssymb}
\usepackage{amsmath}
\usepackage{latexsym}
\usepackage{amsthm}
\usepackage{eepic}
\usepackage{sectsty}
\usepackage{graphicx}

\newtheorem{theorem}{Theorem}
\newtheorem{lemma}[theorem]{Lemma}
\newtheorem{prop}[theorem]{Proposition}

\theoremstyle{definition}
\newtheorem{definition}[theorem]{Definition}

\newtheorem{remark}[theorem]{Remark}

\newtheorem{example}[theorem]{Example}

\newcommand{\microspace}{\mspace{0.5mu}}

\newcommand{\tr}{\operatorname{Tr}}
\newcommand{\pt}{\operatorname{T}}
\renewcommand{\t}{{\scriptscriptstyle\mathsf{T}}}
\newcommand{\ip}[2]{\left\langle #1 , #2\right\rangle}
\def\({\left(}
\def\){\right)}
\def\I{\mathbb{1}}

\newcommand{\setft}[1]{\mathrm{#1}}
\newcommand{\lin}[1]{\setft{L}\left(#1\right)}

\newcommand{\density}[1]{\setft{D}\left(#1\right)}
\newcommand{\unitary}[1]{\setft{U}\left(#1\right)}

\newcommand{\herm}[1]{\setft{Herm}\left(#1\right)}
\newcommand{\pos}[1]{\setft{Pos}\left(#1\right)}

\newcommand{\ppt}[1]{\setft{PPT}\left(#1\right)}

\def\complex{\mathbb{C}}
\def\real{\mathbb{R}}

\def\integer{\mathbb{Z}}

\def \lket {\left|}
\def \rket {\right\rangle}
\def \lbra {\left\langle}
\def \rbra {\right|}
\newcommand{\ket}[1]{\lket\microspace #1 \microspace\rket}
\newcommand{\bra}[1]{\lbra\microspace #1 \microspace\rbra}

\def\X{\mathcal{X}}
\def\Y{\mathcal{Y}}
\def\A{\mathcal{A}}
\def\B{\mathcal{B}}

\usepackage{color}

\begin{document}

\title{\bf PPT-indistinguishable states via semidefinite programming}
 \author{
   {\large Alessandro Cosentino\thanks{acosenti@cs.uwaterloo.ca}}\\[2mm]
   {\it Institute for Quantum Computing and School of Computer Science}\\
   {\it University of Waterloo}
 }
\date{May 4, 2012}

\maketitle

\begin{abstract}
We show a simple semidefinite program whose optimal value is equal to the maximum probability of 
perfectly distinguishing orthogonal maximally entangled states using any PPT measurement 
(a measurement whose operators are positive under partial transpose).
When the states to be distinguished are given by the tensor product of Bell states, 
the semidefinite program simplifies to a linear program.
In [Phys. Rev. Lett. 109, 020506 (2012)], Yu, Duan and Ying exhibit a set of 
$4$ maximally entangled states in $\complex^{4}\otimes\complex^{4}$, 
which is distinguishable by any PPT measurement only with probability strictly less than $1$.
Using semidefinite programming, we show a tight bound of $7/8$ on this probability 
($3/4$ for the case of unambiguous PPT measurements).
We generalize this result by demonstrating a simple construction of a set of $k$ states in $\complex^{k}\otimes\complex^{k}$ 
with the same property, for any $k$ that is a power of $2$.
By running numerical experiments, we obtain some non-trivial results about the PPT-distinguishability of
certain interesting sets of generalized Bell states in $\complex^{5}\otimes\complex^{5}$ and $\complex^{6}\otimes\complex^{6}$.
\end{abstract}

\section{Introduction}
\label{sec:introduction}
A subject of much interest in quantum information theory is understanding powers and limitations of 
the set of quantum operations and measurements defined within the paradigm of \emph{LOCC}, 
short for \emph{Local Operations and Classical Communication}. 
This is aimed at a more general understanding of the role of 
entanglement and non-locality in quantum information.

The problem of distinguishing certain sets of pure states is one of the most basic problems among those 
used to test what can and what cannot be achieved using LOCC protocols.
We will consider this problem in the bipartite case, for which the setup is very simple. 
Suppose that Alice and Bob are given a shared quantum state, 
drawn with some probability from a set of orthogonal states of which they have full knowledge.
Their goal is to determine which state is given. We could consider variants of this problem,
according to how much error we allow, but we will only investigate the case of 
\emph{perfect distinguishability}, where no error is allowed.
The question we are interested in is for what sets of states Alice and Bob are able 
to perfectly achieve their goal by performing only LOCC protocols.
The sets we consider contain only mutually orthogonal states, so the restriction of allowing only LOCC protocols is important.  
If global operations were permitted, Alice and Bob could obviously distinguish the states with no error.

A fundamental result in this area is by Walgate et al. \cite{Walgate00}, 
who established that any two orthogonal pure states can be locally distinguished with no error.
This result has been extended to the case of three maximally entangled states when 
Alice and Bob's systems are $3$-dimensional \cite{Nathanson05}.
Both these results show a surprising power of LOCC protocols.
On the other hand, there exist examples of larger sets that are not perfectly distinguishable 
if we limit the allowed operations to the LOCC framework.
In fact, if both Alice and Bob hold $d$-dimensional systems, 
it is impossible for them to locally distinguish {\em any} $k > d$ 
maximally entangled states \cite{Ghosh04}.
It is important to observe that entanglement is not an essential feature of indistinguishable sets of states. 
For instance, Bennett et al. \cite{Bennett99} showed a set containing only product states 
that cannot be perfectly distinguished by LOCC protocols.

It is natural to ask what is the upper bound on the number of states that can be perfectly distinguishable by LOCC measurements.
If Alice and Bob's systems are $d$-dimensional, is it always possible to locally distinguish a set of $k \leq d$
orthogonal states?
If we allow product states to be in the set, we can easily construct indistinguishable sets with a fixed size 
in any dimension we like, by using the result in \cite{Ghosh04}.
However the question becomes interesting when we consider sets consisting only of maximally entangled states.
In some sense, entanglement makes distinguishability harder, but can also be used as a resource by the two parties.
For $d \leq 3$, the above-mentioned results by Walgate et al. \cite{Walgate00} and Nathanson \cite{Nathanson05} 
give a positive answer to the question. For $d \geq 4$, the problem is not yet as well understood.
Fan \cite{Fan04} showed that when $d$ is prime, any $k$ orthogonal maximally entangled states can be perfectly distinguished
if $k(k-1) \leq 2d$.
Recently, Bandyopadhyay et al. \cite{Ghosh11} gave examples of sets of $k \leq d$ maximally entangled states 
in $\complex^{d} \otimes \complex^{d}$ for $d = 4, 5, 6$ that are not perfectly distinguishable by \emph{one-way} LOCC protocols. 
Interestingly, for $d = 5$ and $d = 6$, they showed sets of size $d-1$.
In another recent result, Yu et al. \cite{Duan11} gave an example of a set of $4$ maximally entangled states in
$\complex^{4} \otimes \complex^{4}$ that cannot be perfectly distinguished by \emph{positive partial transpose} operations (PPT operations).
PPT operations form a superset (in fact, a strict superset) of {\em separable} operations, which, in turns, form a strict superset of 
LOCC operations. Therefore, any upper bound on the power of PPT operations for achieving a particular task holds also against LOCC operations.
The partial transpose mapping has an interesting relationship to entanglement and distillation, with the Peres-Horodecki criterion
being the most renowned application of this relationship. Moreover,
the structure of the set of PPT operations is mathematically simpler than the one of LOCC operations
and many problems are easier to handle when we consider PPT rather than LOCC.
In fact, the set of positive partial transpose operators form a closed convex cone 
and many problems concerning them can be studied by using semidefinite programming, see \cite{Rains00}, for instance.
Yu et al. \cite{Duan11} also noticed that the above-mentioned result about the indistinguishability of any set of $k > d$ 
states (\cite{Ghosh04}) holds even if we broaden the set of allowed operations to PPT.

In this paper we show how the success probability of distinguishing a set of states by using PPT measurements 
can be expressed as the solution of a semidefinite program. A consequence of this is a simpler proof 
using semidefinite duality that the set given in \cite{Duan11} is not perfectly distinguishable by PPT measurements.
In particular we show a tight bound of $7/8$ on the probability of success.
We generalize this result by showing an easy construction of sets with the same properties 
for the case when $k = d$ is any power of $2$.
Another consequence of expressing the problem as a semidefinite program is that for small dimensions 
we can find the optimal solution by running a semidefinite programming solver.
By doing that, we find that the PPT approach cannot be used to prove that the sets mentioned by 
Bandyopadhyay et al. in \cite{Ghosh11} are not perfectly distinguishable by LOCC measurements. 
Recall that they only prove the impossibility of perfectly distinguishing them by one-way LOCC protocols.  
On the other hand, again by numerical calculation, we find examples of local indistinguishable sets of non-trivial size, 
whose states lie in systems of the same dimensions as the ones in \cite{Ghosh11}. 
In particular, these sets were considered in \cite{Ghosh04}, 
where they were shown to be not distinguishable by teleportation protocols,
a subset of LOCC protocols. 

We conclude the paper with a section about unambiguous PPT discrimination. 
Again, we formulate the problem as a semidefinite program and we show a bound of $3/4$ on the success probability of
distinguishing the set given in \cite{Duan11}, when we restrict the strategy to be unambiguous.

\section{Preliminaries}
\label{sec:preliminaries}

Throughout this paper we will use notation and terminology that, for most part,
is standard in quantum information theory.
All vector spaces discussed are assumed to be complex Euclidean spaces.
We write $\lin{\X,\Y}$ to denote the space of linear mappings from a space $\X$ to a space $\Y$, 
and we write $\lin{\X}$ as shorthand for $\lin{\X,\X}$.
For any space $\X$, we write $\herm{\X}$, $\pos{\X}$, $\density{\X}$ and $\unitary{\X}$ to denote the sets
of all Hermitian operators, positive semidefinite operators, density operators and unitary operators on $\X$, respectively.
The identity operator acting on a given space $\X$ is denoted by $\I_{\X}$, or just as $\I$ when $\X$ is implicit.
For Hermitian operators $A, B \in \herm{\X}$ the notations $A\geq B$ 
and $B \leq A$ indicate that $A - B$ is positive semidefinite.
When we refer to a channel, we mean a completely positive, trace-preserving linear mapping of the form
$$
\Phi:\lin{\X} \rightarrow \lin{\Y}.
$$
The \emph{transpose} mapping $T : \lin{\X} \rightarrow \lin{\X}$ is the positive (non completely positive) mapping defined
as $T(X) = X^{\t}$ for all $X \in \lin{\X}$. 
The \emph{partial transpose} on $\X\otimes\Y$ 
is the mapping defined by tensoring the transpose mapping acting on $\X$
and the identity mapping acting on $\Y$ and it is denoted as
\[
 \pt_{\X} = T \otimes \I_{\lin{\Y}}.
\]
Positive operators that remain positive under the action of partial transposition are called \emph{PPT operators}.
We write $\ppt{\X:\Y}$ to denote the set of all PPT operators on a tensor product space $\X\otimes\Y$.
Notice that for the definition of PPT operator, the subspace on which we apply the partial transposition does not matter.
Let us also notice that the set $\ppt{\X:\Y}$ is a closed convex cone.
We will assume that $\A = \complex^{d}$ and $\B = \complex^{d}$ are two identical vector spaces referring to Alice's and Bob's systems respectively.
A pure state $u \in \A\otimes\B$ lying across these spaces is called \emph{maximally entangled} if 
we are left with a maximally mixed state once we trace out one of the spaces, 
i.e., $\tr_{\A}(uu^{*}) = \tr_{\B}(uu^{*}) = \I/d$.

In the rest of the paper we will use the standard Pauli matrices,
$\sigma_{0}=\I, \sigma_{1},\sigma_{2},\sigma_{3} \in \unitary{\complex^{2}}$,
and the standard set of \emph{Bell states} 
$\{ \ket{\psi_{i}} \in \complex^{2}\otimes\complex^{2} :  i\in [0,3] \}$, where
$$
\ket{\psi_{0}} = \frac{1}{\sqrt{2}}(\ket{00}+\ket{11}) \quad\mbox{and}\quad 
\ket{\psi_{i}} = (\I \otimes \sigma_{i})\ket{\psi_{0}}, \quad\mbox{for } i=1,2,3.
$$
For any positive integer $d$, let $\integer_{d}$ be the ring of integers modulo $d$ and $\omega_{d} = \exp(2 \pi i/d)$.
For any choice of $(a,b) \in \integer_{d}^{2}$ we define the \emph{generalized Bell state} 
$\ket{\psi_{a,b}} \in \complex^{d}\otimes\complex^{d}$ as follows:
 \[
    \ket{\psi_{a,b}} = \frac{1}{\sqrt{d}}\sum_{j=0}^{d-1}\omega_{d}^{aj}\ket{j}\otimes\ket{j+b},
 \]
where addition is in $\integer_{d}$.
Whenever we will write states as lowercase Greek letters out of the kets, 
we will mean their density operator representation, for example, $\psi_{0} = \ket{\psi_{0}}\bra{\psi_{0}}$. 
A \emph{measurement} on a space $\X$ is a set of operators $\{P_{a} : a \in \Gamma \} \subset \pos{\X}$, 
indexed by a finite, nonempty set of measurement outcomes $\Gamma$, for which the following constraint holds:
$$
\sum_{a\in\Gamma}P_{a} = \I_{\X}.
$$

In the rest of this paper, we will make use of semidefinite programming.
For a formalization of semidefinite programming similar to the one used in this paper 
and a general overview of semidefinite duality theory, see \cite{Watrous09}, for instance.

\section{PPT distinguishability}
\label{sec:ppt}

Let $\A$ and $\B$ be the complex Euclidean spaces corresponding to Alice and Bob's systems 
and let $S = \{u_1, ... , u_k\} \subset \A\otimes\B$ be a set
of mutually orthogonal unit vectors. Alice and Bob are given a pure state $u_i \in S$, for 
some $i \in \{1,\ldots,k\}$ drawn with some probability $p_i$, and their goal is to determine the value of $i$, 
assuming that they have complete knowledge of the set $S$.

A measurement $\{P_{a} : a \in \Gamma \} \subset \pos{\A\otimes\B}$ is said to be PPT if it can be implemented by a PPT channel, 
or equivalently, if each measurement operator is PPT, that is, $P_{a}\in\ppt{\A:\B}$ for each $a \in \Gamma$.
We say that a set $S$ is {\em PPT-distinguishable} if Alice and Bob can achieve the goal described above for the set $S$ 
without error and by using only PPT measurements. Otherwise we say that the set $S$ is {\em PPT-indistinguishable}.

\subsection{A semidefinite program for the PPT distinguishability problem}
\label{subsec:sdp-ppt}
We will now describe and analyze a semidefinite program whose optimal value is equal to the maximum success probability 
of PPT-distinguishing the set of states that is given as input to the program.

Let $k > 0$ be an integer, $p \in\real^{k}$ a probability vector, and assume that $S = \{ \rho_{i} \in \A \otimes \B : i=1,\ldots,k \}$ 
is the set of state that Alice and Bob are asked to distinguish. Each state $\rho_{i} \in S$ 
is prepared with probability $p(i)$.
We can phrase the maximum probability of successfully distinguishing $S$ with the following semidefinite program
whose constraints characterize the fact that the measurement must be PPT: 

\begin{center}
    \centerline{\underline{Primal problem}}\vspace{-4mm}
    \begin{align}
      \text{maximize:}\quad & \sum_{j = 1}^k p_{j} \ip{P_j}{\rho_{j}}\notag\\
      \text{subject to:}\quad & P_1+ \cdots + P_k = \I_{\A} \otimes \I_{\B},\label{sdp-primal}\\
      & P_1,\ldots,P_k\in\ppt{\A:\B}.\notag
    \end{align}
\end{center}
Since we are interested in perfect distinguishability, in the rest of the paper we will assume, without loss of generality, 
that each state is prepared with uniform probability, i.e., $p(i) = 1/k$, for each $i=1,\ldots,k$. 
We obtain the following dual problem by routine calculation:

\begin{center}
    \centerline{\underline{Dual problem}}\vspace{-4mm}
    \begin{align}
      \text{minimize:}\quad & \frac{1}{k}\tr(Y)\notag\\
      \text{subject to:}\quad & Y - \rho_{j} \geq \pt_{\A}(Q_{j}), \quad j=1,\ldots,k \; ,\label{sdp-dual}\\
      & Y\in\herm{\A\otimes\B}, \notag\\
      & Q_{1}, \ldots, Q_{k}\in\pos{\A\otimes\B}.\notag
    \end{align}
\end{center}
If we further constrain the dual problem, by imposing equality instead of inequality constraints 
in the above program, we obtain the following version:
\begin{center}
    \begin{align}
      \text{minimize:}\quad & \frac{1}{k}\tr(Y)\notag\\
      \text{subject to:}\quad & Y \geq \pt_{\A}(\rho_{j}), \quad j=1,\ldots,k \; ,\label{sdp-dual-moreconstrained}\\
      & Y\in\herm{\A\otimes\B}\notag.
    \end{align}
\end{center}
Let $\alpha$, $\beta$ and $\beta'$ be respectively the solutions of the primal (\ref{sdp-primal}), 
the dual (\ref{sdp-dual}) and the more constrained dual problem (\ref{sdp-dual-moreconstrained}). 
By the weak duality theorem, we have that $\alpha \leq \beta \leq \beta'$, 
that is, any feasible solution to (\ref{sdp-dual-moreconstrained}) upper-bounds the success probability of 
distinguishing the set of states $\{\rho_{1},\ldots,\rho_{k}\}$ by performing only PPT measurements.

An immediate application of this is the following simple proof of the fact shown in \cite{Duan11} 
that it is impossible for Alice and Bob to perfectly distinguish any set of $k > d$ maximally entangled states
in $\complex^{d}\otimes\complex^{d}$ using only PPT measurements.
\begin{theorem}
No PPT measurement can perfectly distinguish more than $d$ 
maximally entangled states in $\complex^{d}\otimes\complex^{d}$. 
\end{theorem}
\begin{proof}
Let $\A = \B = \complex^{d}$. We are assuming that the states we want to distinguish 
$\{\rho_{1},\ldots,\rho_{k}\} \subset \density{\A\otimes\B}$ 
are all maximally entangled. Then, for each $j = 1, \ldots, k$, we have:
\[
 \pt_{\A}(\rho_{j}) = \frac{1}{d}U_{j} \, ,
\]
for some Hermitian unitary operator $U_{j} \in \unitary{\A\otimes\B} \cap \herm{\A\otimes\B}$.
It holds that $(\I_{\A}\otimes\I_{\B}) \geq U_{j}$, for each $j = 1, \ldots, k$. 
Therefore $Y = (\I_{\A}\otimes\I_{\B})/d$ is a feasible solution of the semidefinite program (\ref{sdp-dual-moreconstrained}) and, 
for any measurement $\{P_{1}, \ldots, P_{k}\} \subset \ppt{\A:\B}$, we have
\[
 \frac{1}{k}\sum_{j = 1}^k \ip{P_j}{\rho_{j}} \leq \frac{1}{k}\tr(Y) = \frac{d}{k}.
\]
\end{proof}

\subsection{Bell diagonal states}
\label{subsec:bell}
The following two basic propositions about Bell states will be used throughout the paper and can be proved by direct inspection.
\begin{prop}\label{transposebell}
Let $\A = \B = \complex^{2}$ and let $\psi_{i} = \ket{\psi_{i}}\bra{\psi_{i}} \in \density{\A\otimes\B}$,
for $i\in [0,3]$, be the density operators corresponding to the standard Bell states.
Then the following equations hold:
\[
\pt_{\A}(\psi_{0}) = \frac{1}{2}\I - \psi_{2}, \quad \pt_{\A}(\psi_{1}) = \frac{1}{2}\I - \psi_{3},
\quad \pt_{\A}(\psi_{2}) = \frac{1}{2}\I - \psi_{0}, \quad \pt_{\A}(\psi_{3}) = \frac{1}{2}\I - \psi_{1}.
\]
\end{prop}
\begin{prop}\label{groupG}
The Bell states are invariant under the following group of local symmetries:
$$
G = \{ \I \otimes \I, \sigma_{1} \otimes \sigma_{1}, \sigma_{2} \otimes \sigma_{2}, \sigma_{3} \otimes \sigma_{3}  \},
$$
i.e., $\psi_{i} = U\psi_{i}U^{*}$ for any $U \in G$ and $i\in [0,3]$.
\end{prop}
\begin{definition}
We will describe the mapping of Proposition \ref{transposebell} with the following 
bijection $f:[0,3]\rightarrow[0,3]$ between indices of the set of Bell states:
\[
 f(0) = 2,\; f(1) = 3,\; f(2) = 0,\; f(3) = 1. 
\]
\end{definition}
Let $v \in \integer_{4}^{t} $ be a $t$-dimensional vector and let 
$\ket{\psi_{v}} \in \complex^{2^{t}}\otimes\complex^{2^{t}}$ 
be the maximally entangled state given by the tensor product of Bell states indexed by the vector 
$v = (v_{1}, \ldots, v_{t})$, that is,
$$
\ket{\psi_{v}} = \ket{\psi_{v_{1}}}\otimes\ldots\otimes\ket{\psi_{v_{t}}}.
$$
In the literature, operators diagonal in the basis $\{\psi_{v} = \ket{\psi_{v}}\bra{\psi_{v}} : v \in \integer_{4}^{t}\}$
are called \emph{lattice operators}, or \emph{lattice states} if they are also density operators \cite{Piani04}.
It turns out that in the case when the set to distinguish contains only lattice states,
the semidefinite program (\ref{sdp-primal}) simplifies remarkably, as the following theorem states.
\begin{theorem}
If $\rho_1, \ldots, \rho_k$ are lattice states, then the probability of successfully PPT-distinguishing them
can be expressed as the optimal value of a linear program.
\end{theorem}
\begin{proof}
We will prove that for any feasible solution of the semidefinite program (\ref{sdp-primal}), 
there is another feasible solution of (\ref{sdp-primal}) consisting only of lattice operators 
for which the objective function takes the same value.
Let $\Delta : \lin{\complex^{2}\otimes\complex^{2}} \rightarrow \lin{\complex^{2}\otimes\complex^{2}}$ 
be the channel defined as follows:
\[
  \Delta(X) = \frac{1}{|G|}\sum_{U \in G} UXU^{*} \, ,
\]
where $G$ is the group of local unitaries defined in Proposition \ref{groupG}. 
The channel $\Delta(X)$ acts on $X$ as a completely dephasing channel in the Bell basis.
Suppose that $\A = \B = \complex^{2^{t}}$ and $\rho_1, \ldots, \rho_k \in \density{\A\otimes\B}$ are lattice states.
We let $\Phi = \Delta^{\otimes t}$ and have
\[
  \ip{P_{j}}{\rho_{j}} = \ip{P_{j}}{\Phi(\rho_{j})} = 
  \ip{\Phi(P_{j})}{\rho_{j}} ,
\]
for any $j = 1, \ldots, k$. 
The channel $\Phi$ is unital and, in fact, it is a mixed unitary channel. 
Therefore, if $P_{1}, \ldots, P_{k}$ are such that $P_{1} + \ldots + P_{k} = \I$, 
then it holds that $\Delta(P_{1}) + \ldots + \Delta(P_{k}) = \I$. 
From the positivity of $\Delta$, we have that $\Phi(P) \geq 0$ for any $P \geq 0$. 
Now we show that the partial transpose mapping commutes with the channel $\Delta$.
First we observe how the partial transposition modifies the action of local operators.
Given $U_{1}\in \lin{\A}$, $U_{2} \in \lin{\B}$ and $X \in \lin{\A \otimes \B}$, we have
\[
 \pt_{\A} [(U_{1} \otimes U_{2})X(U_{1} \otimes U_{2})^{*}] = 
 (\overline{U_{1}}\otimes U_{2})\pt_{\A}(X)(\overline{U_{1}}\otimes U_{2})^{*}.
\]
Notice that for the Pauli matrices, we have $\overline{\sigma_{j}} = \sigma_{j}$ 
for $j \in \{ 0,1,3\}$ and $\overline{\sigma_{2}} = -\sigma_{2}$. 
Therefore 
\[
\Delta(\pt_{\A}(X)) = \pt_{\A}(\Delta(X))\, , \quad \text{for any $X \in \lin{\A\otimes\B}$.}
\]
This observation, along with the positivity of $\Delta$, leads to the following implication, which concludes the proof:
\[
 \pt_{\A}(X) \geq 0 \Rightarrow \Delta(\pt_{\A}(X)) \geq 0 \Rightarrow \pt_{\A}(\Delta(X)) \geq 0\, , \quad 
 \text{for any $X \in \lin{\A\otimes\B}$.}
\]
\end{proof}
Even though the states we will consider in the next sections are lattice states, 
we will always refer to the more general semidefinite programming formulation given 
in Section \ref{subsec:sdp-ppt}, rather than expressing the problem in a more explicit linear programming form. 

\subsection{Examples of indistinguishable sets}
\label{sec:examplesec}
We are now ready to show some sets of $k$ maximally entangled states in $\complex^{k}\otimes\complex^{k}$ 
that are not distinguishable by PPT measurements.
\subsubsection{$k = d = 4$}
\label{sec:example}
The following set of $k = 4$ maximally entangled states was shown in \cite{Duan11}
to be not perfectly distinguishable by PPT measurements.
Here we prove via semidefinite programming that the optimal probability of success 
of distinguishing this set for any PPT measurement is $7/8$.
We do this by exhibiting a feasible solution of the more constrained dual problem (\ref{sdp-dual-moreconstrained})
for which the objective function has value $7/8$.

\begin{example}
\label{duanexample}
Let $\A = \A_{1}\otimes\A_{2}$ and $\B = \B_{1}\otimes\B_{2}$ be respectively Alice and Bob's system, 
with $\A_{1} = \A_{2} = \B_{1} = \B_{2} = \complex^{2}$.
The set considered in \cite{Duan11} is 
$\{ \rho_{i} = \ket{x_{i}}\bra{x_{i}} \, : \, i\in [1,4]\}$, where 
\begin{align*}
 \ket{x_{1}} &= \ket{\psi_{0}}\otimes\ket{\psi_{0}}, \\
 \ket{x_{2}} &= \ket{\psi_{1}}\otimes\ket{\psi_{3}}, \\
 \ket{x_{3}} &= \ket{\psi_{2}}\otimes\ket{\psi_{3}}, \\ 
 \ket{x_{4}} &= \ket{\psi_{3}}\otimes\ket{\psi_{3}},
\end{align*}
and the bipartition is such that $\ket{x_{i}} \in \A_{1}\otimes\B_{1}\otimes\A_{2}\otimes\B_{2}$ for each $i\in [1,4]$.  
\end{example}

\begin{theorem}
\label{bound_example}
The maximal probability of success of distinguishing the set of Example (\ref{duanexample}) 
with a PPT measurement is equal to $7/8$.
\end{theorem}
\begin{proof}
It is easy to check that the following operator satisfies the constraints in (\ref{sdp-dual-moreconstrained}) 
and its trace is equal to $7/2$:
\[
Y = \frac{1}{4}\I\otimes\I - \frac{1}{2}(\psi_{2}\otimes\psi_{1}).
\]
We will check the constraint $Y \geq \pt_{\A}(\rho_1)$ and the reader can check the remaining constraints with
a similar calculation. By Proposition \ref{transposebell}, we have
\begin{align*}
 \pt_{\A}(\rho_{1}) &= \pt_{\A}(\psi_{0}\otimes\psi_{0}) = (\frac{1}{2}\I - \psi_{2})\otimes(\frac{1}{2}\I - \psi_{2}) \\
&= \frac{1}{4}\I\otimes\I - \frac{1}{2}\sum_{i\in\{0,1,3\}}(\psi_{i}\otimes\psi_{2}+\psi_{2}\otimes\psi_{i})
\end{align*}
and
\begin{align*}
 Y - \pt_{\A}(\rho_{1}) &= \frac{1}{2}(\psi_{0}\otimes\psi_{2}+\psi_{1}\otimes\psi_{2}+\psi_{3}\otimes\psi_{2}
+\psi_{2}\otimes\psi_{0}+\psi_{2}\otimes\psi_{3}) \geq 0.
\end{align*}
\end{proof}

\begin{theorem}
\label{solutionprimal}
 The bound of Theorem \ref{bound_example} is tight. In fact there is a PPT measurement that achieves the same value.
\end{theorem}
\begin{proof}
Let $Q \in \pos{\complex^{4}\otimes\complex^{4}}$ and $R,S \in \pos{\complex^{2}\otimes\complex^{2}}$ be the following operators:
\[
Q = \frac{1}{4}\I\otimes(\psi_{1}+\psi_{2}), \qquad
R = \frac{7}{8}\psi_{0}+\frac{1}{8}\psi_{3}, \qquad 
S = \frac{1}{8}\psi_{0}+\frac{7}{8}\psi_{3}.
\]
Then the following operators define a PPT measurement that distinguishes the set of Example \ref{duanexample} 
with success probability $7/8$:
\begin{align*}
  P_{1} &= Q + (\frac{2}{3}\psi_{0}+\frac{1}{3}\I)\otimes R ,\\
  P_{2} &= Q + (\frac{1}{3}\psi_{0}+\psi_{1})\otimes S + \frac{1}{3}(\psi_{2}+\psi_{3})\otimes R ,\\
  P_{3} &= Q + (\frac{1}{3}\psi_{0}+\psi_{2})\otimes S + \frac{1}{3}(\psi_{1}+\psi_{3})\otimes R ,\\
  P_{4} &= Q + (\frac{1}{3}\psi_{0}+\psi_{3})\otimes S + \frac{1}{3}(\psi_{1}+\psi_{2})\otimes R .\\
\end{align*}
It is easy to check that these operators define a measurement, 
that is $\sum_{i=1}^{4}P_{i}=\I$.
Using the equations of Proposition (\ref{transposebell}) it is easy to check that those operators are also PPT.
For instance, 
\[
\pt_{\A}(P_{1}) = 
(\psi_{1} + \psi_{2} + \psi_{4})\otimes(\frac{1}{3}\psi_{1}+\frac{1}{2}\psi_{2}+\frac{1}{3}\psi_{4})
+ \frac{1}{4}\psi_{3}\otimes(\psi_{2}+\psi_{3}) \geq 0. 
\]
Finally, we have that $\ip{P_{i}}{\rho_{i}} = \frac{7}{8}$ for each $i\in [1,4]$.
\end{proof}
\subsubsection{$k = d = 2^{n}$, $n \geq 2$}
In \cite{Ghosh11}, the authors pose the question of whether there exists a set of $k$ maximally entangled states 
in $\complex^{d}\otimes\complex^{d}$ not perfectly distinguishable by LOCC, for some $k$ such that $4 < k \leq d$. 
Here we give an explicit construction of such sets when $k=d$ is any power of $2$ and the states are given
by the tensor product of Bell states.
\begin{lemma}
\label{lemma:xor}
Given a vector $\mathbf{i} = (i_{1}, \ldots, i_{n})$, we define the set
\[
S(\mathbf{i}) = \{(j_{1},\ldots,j_{n}) : \bigoplus_{l=1,\ldots,n}\delta_{i_{l}j_{l}} = 1\}, 
\]
where $\oplus$ denotes the sum modulo $2$ and $\delta_{i_{l}j_{l}} = 1$ if and only if $i_{l} = j_{l}$.
Let $\A = \B = \complex^{2^n}$.
Then the partial transpose of $\psi_{\mathbf{i}} = \psi_{i_{1}}\otimes\ldots\otimes\psi_{i_{n}} \in \density{\A\otimes\B}$ is equal to
\[
 \pt_{\A}(\psi_{\mathbf{i}}) = \frac{1}{2^{n}}\left(\I - 2\sum_{(j_{1},\ldots,j_{n}) \in S(\mathbf{i})}\psi_{f(j_{1})}\otimes\cdots\otimes\psi_{f(j_{n})}\right),
\]
where $f$ is the bijection defined in Section \ref{subsec:bell}.
\end{lemma}
\begin{proof}
It follows straightforwardly from Proposition \ref{transposebell}.
\end{proof}
\begin{theorem}
For any $n \geq 2$, there is a set of $k = 2^{n}$ maximally entangled states in $\complex^{k}\otimes\complex^{k}$ 
that is not perfectly distinguishable by PPT operations.
\end{theorem}
\begin{proof}
The case $n = 2$ is covered in Section \ref{sec:example}.
Here we construct a set for any $n \geq 3$.
Consider the set of states $S = \{ \rho_{j} = \psi_{0}\otimes\rho_{j}' : j=1,\ldots,k \}$, 
where each $\rho_{j}'$ is a tensor product of one of the $3^{n-1}$ combinations 
of Bell states different from $\psi_{0}$.
Since $3^{n-1} > 2^{n}$ for any $n \geq 3$, we can always construct such set. 
By using Lemma \ref{lemma:xor}, it is easy to check that the operator 
$$
Y = \frac{1}{k}\I - \frac{2}{k}\psi_{2}\otimes\ldots\otimes\psi_{2}
$$
satisfies the constraints of the semidefinite program (\ref{sdp-dual-moreconstrained}). Also, its trace is strictly less than $1$. 
\end{proof}

\begin{remark}
An interesting feature of the set of states $S$ considered in the above proof is that Alice and Bob 
are basically being provided with the maximally entangled pair $\psi_{0} \in \density{\complex^{2}\otimes\complex^{2}}$ as a resource, 
but they are still not able to distinguish the set $\{\rho_{j}' : 1 \leq j \leq 2^{n} \} \subset \density{\complex^{2^{n-1}}\otimes\complex^{2^{n-1}}}$. 
In fact, for larger $n$, it is easy to see that we can even give them $c > 1$ entangled pairs, 
as long as $c$ is odd and $c \leq (1-\log_{3}2)n$, and a construction of an indistinguishable set similar to the one above will still work.
\end{remark}

\begin{remark}
The upper bound of the probability of distinguishing the sets we derive from the semidefinite program is $1-2/k^{2}$, 
which is the value of the trace of the operator $Y$ in the above theorem multiplied by $1/k$.
Notice that there exist sets that are in some sense even more indistinguishable.
For example, in the case of $k=8$, 
we could show that the following set of states cannot be PPT-distinguished with a success 
probability bigger than $15/16$:
\[
\{ \psi_{(1,1,1)}, \psi_{(1,1,3)}, \psi_{(1,1,4)}, \psi_{(2,2,2)},
\psi_{(3,3,1)}, \psi_{(3,3,3)}, \psi_{(3,3,4)}, \psi_{(4,2,2)} \},
\]
where $\psi_{(i,j,k)} = \psi_{i}\otimes\psi_{j}\otimes\psi_{k}$.
\end{remark}

\subsubsection{$k=d=5,6$}
\label{subsec:fiveandsix}
We ran the semidefinite programming solver \texttt{CVX} \cite{CVX} against the sets of the examples given 
in \cite{Ghosh11} for the case $k=d=5,6$ 
and they turned out to be perfectly distinguishable by PPT measurements.
Therefore the question they pose, whether LOCC protocols more powerful than one-way protocols can perfectly distinguish those sets, 
remains open.

On the other hand, again by running numerical computations with \texttt{CVX}, 
we show that the following two sets, respectively of $k=5$ and $k=6$ generalized Bell states in $\complex^{k}\otimes\complex^{k}$, 
are not perfectly distinguishable by PPT measurements. 
These examples come from \cite{Ghosh04}, where it was shown that cannot be reliably 
distinguished by so-called standard teleportation protocols, which are a subset of LOCC protocols.

\begin{example}[$d = k = 5$]
Any PPT measurement errs with probability at least $0.0101$ when trying to distinguish the set of generalized Bell states
\[
\psi_{0,0}, \psi_{1,1}, \psi_{2,1}, \psi_{1,3}, \psi_{2,3} \in \complex^{5}\otimes\complex^{5}. 
\]
\end{example}

\begin{example}[$d = k = 6$] 
Any PPT measurement errs with probability at least $0.002$ when trying to distinguish the set of generalized Bell states
\[
\psi_{0,0}, \psi_{1,0}, \psi_{2,0}, \psi_{3,0}, \psi_{4,0}, \psi_{0,3} \in \complex^{6}\otimes\complex^{6}.
\]
\end{example}

\section{Unambiguous PPT discrimination}
\label{sec:unambiguous}
In the previous section, we analyzed the problem of distinguishing quantum states 
using PPT measurements that minimize the probability of error.
Bandyopadhyay \cite{Som} raised the question of what is the probability of error if, instead,
we consider an \emph{unambiguous} PPT strategy to distinguish the sets of states of Section \ref{sec:examplesec}.
In such strategy, Alice and Bob never give an incorrect answer, although their answer can be inconclusive.
If there are $k$ states to be distinguished, an unambiguous measurement consists of $k+1$ operators, 
where the outcome of the operator $P_{k+1}$ corresponds to the inconclusive answer.
In this section, we cast this problem into the framework of semidefinite programming and we make a comparison 
with the result we obtained in Section \ref{sec:ppt} for the example considered in \cite{Duan11}. 
The semidefinite programming approach has already been used to study unambiguous discrimination \cite{Eldar03}, 
but never, as far as we know, to study unambiguous PPT discrimination. 
In fact, we believe that unambiguous PPT discrimination in general, or even unambiguous LOCC discrimination,
has not been thoroughly investigated yet.

The optimal value of the following semidefinite program is equal to the success probability of 
unambiguously distinguishing a set of states $\{ \rho_{1}, \ldots, \rho_{k} \}$ using PPT measurements.
Again, we assume that the states are drawn with a uniform probability.
\begin{center}
    \centerline{\underline{Primal problem}}\vspace{-4mm}
    \begin{align}
      \text{maximize:}\quad & \frac{1}{k} \sum_{j = 1}^k \ip{P_j}{\rho_{j}}\notag\\
      \text{subject to:}\quad & P_1+ \cdots + P_{k+1} = \I_{\A} \otimes \I_{\B}, \label{sdp-primal-unambiguous}\\
      & P_1,\ldots,P_{k+1}\in\ppt{\A:\B}, \notag\\
      & \ip{P_{i}}{\rho_{j}} = 0, \qquad 1 \leq i,j \leq k, \quad i \neq j. \notag
    \end{align}
\end{center}

\begin{center}
    \centerline{\underline{Dual problem}}\vspace{-4mm}
    \begin{align}
      \text{minimize:}\quad & \frac{1}{k}\tr(Y)\notag\\
      \text{subject to:}\quad & Y - \rho_{j} + \sum_{\substack{1\leq i \leq k \\ i\neq j}}y_{i,j}\rho_{i} \geq \pt_{\A}(Q_{j}), \quad j=1,\ldots,k \; ,\notag\\
      & Y \geq \pt_{\A}(Q_{k+1}),\label{sdp-dual-unambiguous}\\
      & Q_{1}, \ldots, Q_{k+1}\in\pos{\A\otimes\B},\notag\\
      & Y \in \herm{\A\otimes\B},\notag\\
      & y_{i,j} \in \real, \quad 1 \leq i,j \leq k, \quad i \neq j.\notag
    \end{align}
\end{center}
Interestingly, the optimal probability of unambiguously distinguish the set of states of Example \ref{duanexample} 
with PPT measurements is $3/4$, which should be compared with the success probability of $7/8$ that can be achieved 
with a minimum-error strategy (see Theorem \ref{solutionprimal}). 
In fact, using a semidefinite program solver, we were also able to verify that this bound is actually tight.

\begin{theorem}
The maximum success probability of unambiguously distinguish the set of states of Example \ref{duanexample}
with PPT measurements is equal to $3/4$. 
\end{theorem}
\begin{proof}
We show a feasible solution of the dual problem for which the value of the objective function is $3/4$. Let
\[
Y = \frac{1}{4}[(\I-\psi_{1})\otimes(\I - 2\psi_{4}) + \psi_{1}\otimes(-\psi_{1}+3\psi_{2}+3\psi_{3}+\psi_{4})].
\]
and
\begin{align*}
 Q_{1} &= (\I - \psi_{3})\otimes\psi_{3} + \psi_{3}\otimes(\psi_{2}+\psi_{3}), \\
 Q_{2} &= (\psi_{1} + \psi_{2})\otimes\psi_{2} + \psi_{4}\otimes(\I - \psi_{2}), \\
 Q_{3} &= (\psi_{2} + \psi_{4})\otimes\psi_{2} + \psi_{1}\otimes(\I - \psi_{2}), \\
 Q_{4} &= (\psi_{1} + \psi_{4})\otimes\psi_{2} + \psi_{2}\otimes(\I - \psi_{2}), \\
 Q_{5} &= \psi_{3}\otimes\psi_{2}.
\end{align*}
We can use Proposition \ref{transposebell} to check that the following equations hold:
\[
Y - \rho_{j} + \sum_{\substack{1\leq i \leq k \\ i\neq j}}\rho_{i} = \pt_{\A}(Q_{j}), \quad j=1,\ldots,4 \quad\text{and}\quad Y \geq \pt_{\A}(Q_{5}),
\]
i.e., the constraints of the program (\ref{sdp-dual-unambiguous}) are satisfied. Also, we have that $\tr(Y) = 3$.
\end{proof}

\section{Conclusion}
In summary, we have extended previously known results about the indistinguishability 
of some sets of $d$ orthogonal maximally entangled states in $\complex^{d}\otimes\complex^{d}$ using PPT measurements.
We hope that our approach based on semidefinite programming
can lead to a better understanding of the power of local operations 
at least for what concerns the task of distinguishing quantum states.

The main open question is whether there exist examples of sets such as the ones we considered, but of size $k < d$. 
For small values of $d$, we obtained an unsuccessful answer to this question 
from running an exhaustive numerical search against sets consisting of $d-1$ 
generalized Bell states, or lattice states, or states constructed from complex Hadamard matrices. 
It would be interesting to try some different constructions of orthogonal maximally entangled states.
The following are other interesting unanswered questions related to our results.
\begin{itemize}
  \item Can we achieve the same bound of Theorem \ref{solutionprimal} using a separable (or LOCC) measurement?
  \item Are there examples of sets for which LOCC measurements do worse than PPT?
\end{itemize}

\subsection*{Acknowledgments}
I am grateful to John Watrous for sharing the idea that inspired this paper and for several useful discussions.
I would like to thank Sevag Gharibian and Nengkun Yu for helpful email exchanges,
as well as Marco Piani for pointing out to me the definition of lattice states in \cite{Piani04}, and 
Som Bandyopadhyay, whose observations led me to the result of Section \ref{sec:unambiguous}.



\end{document}